\documentclass[11pt,a4paper]{article}
\usepackage[T1]{fontenc}
\usepackage[utf8]{inputenc}
\usepackage{lmodern}
\usepackage[margin=28mm]{geometry}

\usepackage{amsmath,amssymb,amsthm,mathtools}
\usepackage{microtype}
\usepackage{aliascnt}
\usepackage{needspace}
\usepackage{booktabs}
\usepackage{flafter}
\usepackage{placeins}

\usepackage{algorithm}
\usepackage[noend]{algpseudocode}
\usepackage{xcolor}
\usepackage{tikz}
\usetikzlibrary{arrows.meta}

\usepackage{xurl}
\usepackage{hyperref}
\def\HYPERCOLOR{blue!70!black}
\hypersetup{colorlinks=true,linkcolor=\HYPERCOLOR,citecolor=\HYPERCOLOR,urlcolor=\HYPERCOLOR}
\usepackage[capitalize,nameinlink,noabbrev]{cleveref}

\makeatletter
\providecommand{\theHALG@line}{}
\renewcommand{\theHALG@line}{\thealgorithm.\arabic{ALG@line}}
\makeatother

\newtheorem{theorem}{Theorem}[section]
\newaliascnt{lemma}{theorem}
\newtheorem{lemma}[lemma]{Lemma}
\aliascntresetthe{lemma}
\crefname{lemma}{Lemma}{Lemmas}
\Crefname{lemma}{Lemma}{Lemmas}
\AddToHook{env/theorem/begin}{\crefalias{section}{theorem}}
\AddToHook{env/lemma/begin}{\crefalias{section}{lemma}}

\crefname{algorithm}{Algorithm}{Algorithms}
\Crefname{algorithm}{Algorithm}{Algorithms}
\algrenewcommand\algorithmiccomment[1]{\hfill\(\triangleright\) #1}
\newcommand{\inputsize}{n}
\newcommand{\outputsize}{k}
\newcommand{\seqval}{a}
\newcommand{\weight}{w}
\newcommand{\rangeopt}{F}
\newcommand{\dpopt}{D}
\newcommand{\envelope}{H}
\newcommand{\predenv}{P}
\newcommand{\candidateset}{E}
\newcommand{\lowset}{L}
\newcommand{\highset}{R}

\newcommand{\noderoot}{T}
\newcommand{\nodeinterval}{I}
\newcommand{\outertree}{\mathsf{outer}}
\newcommand{\rootptr}{\mathsf{root}}
\newcommand{\witnessU}{U}
\newcommand{\witnessV}{V}
\newcommand{\suffix}{S}
\newcommand{\witnessweight}{W}
\newcommand{\emptyseq}{\varepsilon}
\newcommand{\scalar}{x}

\newcommand{\idx}{i}
\newcommand{\jdx}{j}
\newcommand{\leftbound}{l}
\newcommand{\rightbound}{r}
\newcommand{\pointp}{p}
\newcommand{\pointq}{q}
\newcommand{\blockleft}{b}
\newcommand{\blockright}{e}
\newcommand{\cut}{c}
\newcommand{\midpoint}{m}
\newcommand{\funleft}{f}
\newcommand{\funright}{g}
\newcommand{\treenode}{v}
\newcommand{\sourcept}{s}
\newcommand{\terminalpt}{t}
\newcommand{\pathheight}{h}
\newcommand{\tradeoff}{\alpha}

\newcommand{\Children}{\mathsf{Children}}
\newcommand{\JoinTrees}{\mathsf{Join}}
\newcommand{\Merge}{\mathsf{Merge}}
\newcommand{\last}{\mathsf{last}}
\newcommand{\Append}{\mathsf{Append}}
\newcommand{\PrefixAppend}{\mathsf{PrefixAppend}}
\newcommand{\Cat}{\mathsf{Cat}}
\newcommand{\Leaf}{\mathsf{Leaf}}

\title{Near-Linear Preprocessing for Range Weighted Longest Increasing Subsequence Queries}
\author{Keisuke Korogi\\[2pt]
  {\small Yokohama National University, Japan}\\
  {\small
    \href{mailto:keisuke.korogi.52@gmail.com}
         {\nolinkurl{keisuke.korogi.52@gmail.com}}
  }
}
\date{}

\hypersetup{
pdftitle={Near-Linear Preprocessing for Range Weighted Longest Increasing Subsequence Queries},
pdfauthor={Keisuke Korogi},
pdfsubject={Range Weighted Longest Increasing Subsequence Queries},
pdfkeywords={longest increasing subsequence, range queries, persistent data structures, single-crossing property}
}

\begin{document}

\maketitle

\begin{abstract}
We give a deterministic data structure for range weighted longest increasing subsequence (LIS) queries on $n$ values with positive weights. 
Preprocessing takes $O(n\log^2 n)$ time and space. 
Value queries take $O(\log n)$ time, and an optimal subsequence of $k$ elements can be reported in $O(k + \log n)$ time. 
This resolves an open question on truly subquadratic preprocessing with polylogarithmic query time. 
The key is a single-crossing property of envelopes for value-separated endpoint sets, enabling $O(\log n)$-time persistent merges. 
The bounds assume constant-time exact addition and comparison of weights. 
\end{abstract}

\section{Introduction}
\label{sec:introduction}

Let $\seqval_0, \dots, \seqval_{\inputsize-1}$ be a sequence with positive real weights $\weight_0, \dots, \weight_{\inputsize-1}$. 
A \emph{range weighted longest increasing subsequence} query asks for the maximum total weight of a strictly increasing subsequence in an index interval $[\leftbound, \rightbound)$. 
A reporting query also returns an optimal subsequence. 

For unweighted range LIS, Tiskin's semi-local comparison techniques give $O(\inputsize\log^2 \inputsize)$ preprocessing and $O(\log \inputsize)$ time to return the length~\cite{Tiskin15}; see also the formulation in~\cite[Lemma~B.14]{GGK25}. 
Gawrychowski, Gorbachev, and Kociumaka~\cite[Theorem~1.3]{GGK25} support reporting in $O(\outputsize\log^{\tradeoff} \inputsize)$ time after $O(\inputsize\log^{3-\tradeoff}\inputsize)$ preprocessing, for every $\tradeoff \in [0, 1]$. 
Here $\outputsize$ is the output length. 
Their construction uses core-sparse Monge matrix multiplication, where the \emph{core} records nonzero mixed second differences. 
The theorem and appendix numbers cited from~\cite{GGK25} refer to its full version. 

For arbitrary positive weights, Gawrychowski et al.~\cite[Appendix~B.1]{GGK25} ask whether truly subquadratic preprocessing and polylogarithmic query time are possible. 
Their alignment-matrix construction remains correct, but the unweighted core-size guarantee no longer follows from their analysis. 
Karthik C.~S. and Rahul~\cite[Section~7 of the full version]{KR26} also ask for a truly subquadratic algorithm for weighted one-dimensional range LIS.

\Needspace{8\baselineskip}
\subsection{Our result}
We answer these questions without bounding the core size.

\begin{theorem}
\label{thm:main}
There is a deterministic data structure for range weighted LIS with $O(\inputsize\log^2 \inputsize)$ preprocessing time and space, $O(\log \inputsize)$ value-query time, and $O(\outputsize + \log \inputsize)$ reporting time. 
\end{theorem}

The query intervals need not be known during preprocessing. 
The algorithm is exact, and its time and space bounds are independent of the weight magnitudes. 
We count exact additions and comparisons of weights, comparisons of input values, and operations on $O(\log \inputsize)$-bit indices and pointers at unit cost. 
Space counts real-weight cells and integer words.

\begin{table}[t]
\centering
\small
\setlength{\tabcolsep}{5pt}

\caption{
Selected time bounds for static range LIS. 
The value and reporting structures of~\cite{GGK25} can be combined within the stated preprocessing bounds. 
A dash indicates that no reporting bound is stated in the cited publication. 
}
\label{tab:comparison}

\vspace{4mm}

\begin{tabular}{@{}llccc@{}}
\toprule
Result & Weights & Preprocessing & Value query & Reporting \\
\midrule
Tiskin~\cite{Tiskin15} & Unit & $O(\inputsize\log^2 \inputsize)$ & $O(\log \inputsize)$ & --- \\
\cite{GGK25}, $\tradeoff=1$ & Unit & $O(\inputsize\log^2 \inputsize)$ & $O(\log \inputsize)$ & $O(\outputsize\log \inputsize)$ \\
\cite{GGK25}, $\tradeoff=0$ & Unit & $O(\inputsize\log^3 \inputsize)$ & $O(\log \inputsize)$ & $O(\outputsize)$ \\
\Cref{thm:main} & Positive real & $O(\inputsize\log^2 \inputsize)$ & $O(\log \inputsize)$ & $O(\outputsize + \log \inputsize)$ \\
\bottomrule
\end{tabular}

\end{table}

\paragraph{Unweighted specialization. }
Setting all weights to one gives ordinary range LIS. 
At the same $O(\inputsize \log^2 \inputsize)$ preprocessing time, our reporting bound is $O(\outputsize + \log \inputsize)$, compared with $O(\outputsize \log \inputsize)$ in~\cite{GGK25}. 
Their $O(\outputsize)$-time reporting bound uses $O(\inputsize \log^3 \inputsize)$ preprocessing. 
Tiskin~\cite[Section~8.1]{Tiskin13} also states that retaining intermediate data permits output-linear traceback for semi-local LCS queries on permutations.

\subsection{Overview}
For each endpoint $\pointp$, let $\dpopt_{\pointp}(\leftbound)$ be the maximum weight of an increasing subsequence ending at $\pointp$ with all indices at least $\leftbound$, or zero if $\leftbound>\pointp$. 
For an endpoint set $\candidateset$, let $\envelope_{\candidateset}$ be the pointwise maximum of zero and its endpoint functions. 
The optimum for $[\leftbound,\rightbound)$ is $\envelope_{[0,\rightbound)}(\leftbound)$. 

The key is a single-crossing property. 
If all endpoint values in $\lowset$ are smaller than those in $\highset$, then $\envelope_{\highset}(\leftbound) > \envelope_{\lowset}(\leftbound)$ holds exactly on a prefix of left boundaries. 
We derive this from a pairwise property using a planar path-exchange argument~\cite{Schmidt98,GGK25}. 
Consequently, $\envelope_{\lowset\cup\highset}$ consists of a prefix of $\envelope_{\highset}$ and a suffix of $\envelope_{\lowset}$. 

We represent each envelope by a persistent segment tree indexed by $\leftbound$. 
Single crossing permits $O(\log\inputsize)$-time merges with $O(\log\inputsize)$ new nodes. 
A mutable segment tree over input values stores these roots. 
A value-prefix query obtains the predecessor envelope, and a persistent prefix update constructs the new endpoint function. 
Updating the outer tree takes $O(\log\inputsize)$ further merges. 
Thus each element costs $O(\log^2\inputsize)$ time and space. 
Saving the envelope root stored at the outer root after each input prefix reduces a range query to one point evaluation. 

For reporting, persistent nodes store shared suffix witnesses represented by binary concatenations. 
Their weights serve as additive lazy tags. 
A query collects these suffixes along one tree path in $O(\log \inputsize)$ time. 
Expanding them takes $O(\outputsize)$ time. 
The construction requires no bound on Monge core size.

\section{Endpoint functions}
\label{sec:endpoints}

We use $[\blockleft,\blockright]$ and $[\blockleft,\blockright)$ for closed and half-open integer intervals. All logarithms are to base two.
We may assume that $\seqval_0,\dots,\seqval_{\inputsize-1}$ is a permutation of $[0,\inputsize)$: sort the pairs $(\seqval_{\idx},-\idx)$ lexicographically and replace each $\seqval_{\idx}$ by its zero-based rank, keeping the notation $\seqval_{\idx}$.
This takes $O(\inputsize\log\inputsize)$ time and preserves every strictly increasing subsequence, including when input values coincide.

For $0\leq\leftbound\leq\rightbound\leq\inputsize$, let
\begin{equation}
\label{eq:objective}
\rangeopt(\leftbound,\rightbound)\coloneqq
\max\left\{\ \sum_{\jdx=1}^{\outputsize}\weight_{\idx_{\jdx}}
\ \middle|\ \ 
\leftbound\leq\idx_1<\dots<\idx_{\outputsize}<\rightbound,
\quad \seqval_{\idx_1}<\dots<\seqval_{\idx_{\outputsize}}\ \right\}.
\end{equation}
The maximum ranges over all lengths $\outputsize$, including the empty subsequence of weight zero. Thus $\rangeopt(\leftbound,\leftbound)=0$, and every nonempty interval has a positive optimum value. 

For an endpoint $\pointp$, define $\dpopt_{\pointp}(\leftbound)$ to be the maximum weight of an increasing subsequence in $[\leftbound, \pointp]$ ending at $\pointp$ when $\leftbound\leq\pointp$, and set $\dpopt_{\pointp}(\leftbound)=0$ when $\leftbound>\pointp$. 
Each $\dpopt_{\pointp}$ is nonincreasing on $[0,\inputsize)$ and positive exactly on $[0,\pointp]$. 

For an endpoint set $\candidateset\subseteq[0,\inputsize)$, define its envelope by
\begin{equation}
\label{eq:envelope}
\envelope_{\candidateset}(\leftbound)
\coloneqq \max\Bigl(\{0\}\cup\{\dpopt_{\pointp}(\leftbound)\mid\pointp\in\candidateset\}\Bigr). 
\end{equation}
In particular, $\envelope_{\emptyset}=0$. 
The set $\candidateset$ restricts only the endpoint of a subsequence. 
Thus, 
\begin{equation}
\label{eq:query}
\rangeopt(\leftbound,\rightbound)=\envelope_{[0,\rightbound)}(\leftbound),
\end{equation}
since every nonempty feasible subsequence has an endpoint before $\rightbound$, whereas endpoints before $\leftbound$ contribute zero.

The predecessor envelope for position $\idx$ is $\predenv_{\idx}\coloneqq \envelope_{\{\jdx<\idx\, \mid\, \seqval_{\jdx}<\seqval_{\idx}\}}$. 
The usual last-element recurrence becomes
\begin{equation}
\label{eq:dp}
\dpopt_{\idx}(\leftbound)=
\begin{cases}
\predenv_{\idx}(\leftbound)+\weight_{\idx},&\leftbound\leq\idx,\\
0,&\leftbound>\idx.
\end{cases}
\end{equation}

\section{Single-crossing property}
\label{sec:crossing}

We show that, when two endpoint sets are separated by value, the higher-value envelope exceeds the lower-value envelope exactly on a prefix of left boundaries. 
Ties favour the lower-value set. 
The proof uses the standard planar path-exchange argument underlying Monge inequalities~\cite[Fact~B.7]{GGK25}. 

\begin{lemma}
\label{lem:pair}
Let $\pointp, \pointq \in [0, \inputsize)$ satisfy $\seqval_{\pointp} < \seqval_{\pointq}$. 
Then there exists $\cut \in [0, \inputsize]$ such that 
\[
\{ \leftbound \in [0, \inputsize) \mid \dpopt_{\pointp}(\leftbound) < \dpopt_{\pointq}(\leftbound) \} = [0, \cut). 
\]
\end{lemma}
\begin{proof}
If $\pointp < \pointq$, any increasing subsequence ending at $\pointp$ can be extended by $\pointq$. 
For $\leftbound \leq \pointp$, 
\[
\dpopt_{\pointq}(\leftbound) \geq \dpopt_{\pointp}(\leftbound) + \weight_{\pointq} > \dpopt_{\pointp}(\leftbound). 
\]
For $\pointp < \leftbound \leq \pointq$, 
\[
\dpopt_{\pointq}(\leftbound) \geq \weight_{\pointq} > 0 = \dpopt_{\pointp}(\leftbound), 
\]
while both functions vanish for $\leftbound > \pointq$. 
Thus the set is $[0,\, \pointq+1)$. 

Suppose instead that $\pointq<\pointp$. 
We first show, for $0 \leq \leftbound < \leftbound' \leq \pointq$, that 
\begin{equation}
\label{eq:exchange}
\dpopt_{\pointp}(\leftbound) + \dpopt_{\pointq}(\leftbound') \leq \dpopt_{\pointq}(\leftbound) + \dpopt_{\pointp}(\leftbound'). 
\end{equation}
Use the standard alignment-grid representation~\cite{Schmidt98,GGK25}, with positive weights on the diagonal edges. 
Consider the directed grid on $[0, n]^2$ with east and north edges of weight zero. 
For every input position $\idx$, add the diagonal edge 
\[
(\idx, \seqval_{\idx}) \longrightarrow (\idx+1, \seqval_{\idx}+1) \quad \text{of weight } \weight_{\idx}.
\]
The resulting graph is planar: diagonals lie inside distinct unit squares and intersect grid edges only at their endpoints. 

Put $\sourcept_{\leftbound} = (\leftbound, 0)$ and $\terminalpt_{\idx} = (\idx+1, \seqval_{\idx}+1)$. 
For $\leftbound \leq \idx$, the diagonal edges of any $\sourcept_{\leftbound}$--$\terminalpt_{\idx}$ path correspond to an increasing subsequence in $[\leftbound, \idx]$ whose values are at most $\seqval_{\idx}$. 
Conversely, every such subsequence determines a path of the same weight by joining its diagonal edges with zero-weight grid paths. 
If such a subsequence omits $\idx$, its indices and values are smaller than those of $\idx$. 
Appending $\idx$ increases its weight. 
Hence the maximum path weight from $\sourcept_{\leftbound}$ to $\terminalpt_{\idx}$ is exactly $\dpopt_{\idx}(\leftbound)$. 

Choose maximum-weight paths from $\sourcept_{\leftbound}$ to $\terminalpt_{\pointp}$ and from $\sourcept_{\leftbound'}$ to $\terminalpt_{\pointq}$. 
Both lie in $[\leftbound,\, \pointp + 1] \times [0,\, \seqval_{\pointq} + 1]$. 
The four terminals $\sourcept_{\leftbound}$, $\sourcept_{\leftbound'}$, $\terminalpt_{\pointp}$, $\terminalpt_{\pointq}$ are distinct and occur in this cyclic order on the boundary of the rectangle. 
Thus the two paths have a common vertex by planarity; see \cref{fig:exchange}. 
Swapping their suffixes at that vertex gives paths from $\sourcept_{\leftbound}$ to $\terminalpt_{\pointq}$ and from $\sourcept_{\leftbound'}$ to $\terminalpt_{\pointp}$. 
The sum of their weights is unchanged. 
Since their weights are at most $\dpopt_{\pointq}(\leftbound)$ and $\dpopt_{\pointp}(\leftbound')$, respectively, \eqref{eq:exchange} follows. 

By \eqref{eq:exchange}, $\dpopt_{\pointq}(\leftbound) - \dpopt_{\pointp}(\leftbound)$ is nonincreasing on $[0, \pointq]$. 
For $\leftbound > \pointq$, $\dpopt_{\pointq}(\leftbound) = 0 \leq \dpopt_{\pointp}(\leftbound)$. 
Hence the set in the lemma is a prefix of $[0, \inputsize)$. 
\end{proof}

\begin{figure}[t]
\centering
\begin{tikzpicture}[x=1.2cm,y=0.8cm,>=Stealth,font=\small]
  \draw[step=1,black!15,very thin] (0,0) grid (4,4);
  \draw[thin] (0,0) rectangle (4,4);

  \coordinate (s) at (0,0);
  \coordinate (sp) at (1,0);
  \coordinate (zp) at (4,2);
  \coordinate (zq) at (3,4);
  \coordinate (v) at (1,1);

  \draw[thick,->]
    (s) -- (1,1) -- (3,1) -- (zp);
  \draw[thick,dashed,->]
    (sp) -- (1,2) -- (2,3) -- (zq);

  \foreach \p in {s,sp,zp,zq,v}
    {\fill (\p) circle (1.8pt);}

  \node[below left] at (s) {$\sourcept_{\leftbound}$};
  \node[below] at (sp) {$\sourcept_{\leftbound'}$};
  \node[right] at (zp) {$\terminalpt_{\pointp}$};
  \node[above] at (zq) {$\terminalpt_{\pointq}$};
  \node[below right] at (v) {$\treenode$};
  \node[below] at (2,-0.65) {index};
  \node[rotate=90] at (-0.65,2) {value};
\end{tikzpicture}
\caption{
Two maximum-weight paths for the sequence $(0, 2, 3, 1)$
with unit weights. 
Here $\leftbound = 0$, $\leftbound' = 1$,
$\pointq = 2$, and $\pointp = 3$. 
Horizontal and vertical edges have weight zero; diagonal edges represent selected elements. 
}
\label{fig:exchange}
\end{figure}
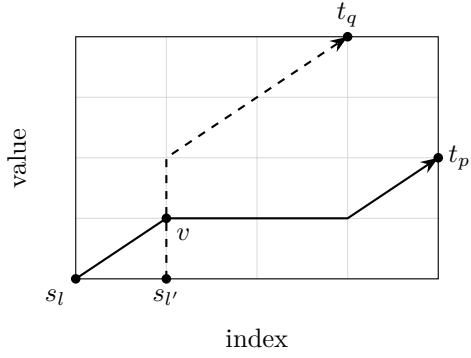

The pairwise property extends to value-separated endpoint sets. 

\begin{lemma}
\label{lem:envelopes}
Let $\lowset,\, \highset \subseteq [0, \inputsize)$ satisfy $\seqval_{\pointp} < \seqval_{\pointq}$ for all $\pointp \in \lowset$ and $\pointq \in \highset$. 
Then there exists $\cut \in [0, \inputsize]$ such that $\envelope_{\lowset}(\leftbound) < \envelope_{\highset}(\leftbound)$ if and only if $\leftbound < \cut$, for every $\leftbound \in [0, \inputsize)$. 
Consequently, 
\begin{equation}
\label{eq:envelope-splice}
\envelope_{\lowset \cup \highset}(\leftbound) = 
\begin{cases}
\envelope_{\highset}(\leftbound), & \leftbound < \cut, \\
\envelope_{\lowset}(\leftbound), & \leftbound \geq \cut. 
\end{cases}
\end{equation}
\end{lemma}
\begin{proof}
Suppose $\envelope_{\lowset}(\leftbound_0) \geq \envelope_{\highset}(\leftbound_0)$. 
If $\envelope_{\lowset}(\leftbound_0) = 0$, both envelopes remain zero for $\leftbound \geq \leftbound_0$. 

Otherwise, choose $\pointp \in \lowset$ attaining $\envelope_{\lowset}(\leftbound_0)$. 
For every $\pointq \in \highset$, 
\[
\dpopt_{\pointq}(\leftbound_0)
\leq \envelope_{\highset}(\leftbound_0)
\leq \dpopt_{\pointp}(\leftbound_0). 
\]
By \cref{lem:pair}, $\dpopt_{\pointq}(\leftbound) \leq \dpopt_{\pointp}(\leftbound)$ for all $\leftbound \geq \leftbound_0$. 
Taking the maximum over $\pointq$ and zero gives 
\[
\envelope_{\highset}(\leftbound)
\leq \dpopt_{\pointp}(\leftbound)
\leq \envelope_{\lowset}(\leftbound)
\qquad (\leftbound \geq \leftbound_0). 
\]
Hence the strict comparison holds on a prefix. 
Choosing $\envelope_{\lowset}$ at equality gives \eqref{eq:envelope-splice}. 
\end{proof}

\section{Persistent envelope representations}
\label{sec:persistence}

We represent envelopes and their witnesses in persistent segment trees. 
All nodes are immutable; updates use path copying~\cite{DSST89}.
Beyond persistent range updates, we only need to merge envelope pairs satisfying the single-crossing property. 

\subsection{Witnesses and suffix tags}
\label{sec:witnesses}

A witness is an index sequence represented by an immutable handle: the empty sequence $\emptyseq$, a singleton $\Leaf(\idx)$, or a binary concatenation $\Cat(\witnessU,\witnessV)$. 
We set $\Cat(\emptyseq,\witnessV)=\witnessV$ and $\Cat(\witnessU,\emptyseq)=\witnessU$; otherwise, concatenation creates one node pointing to the two nonempty handles. 
Each handle stores its total weight $\witnessweight$, where
\[
\witnessweight(\emptyseq)=0,\qquad
\witnessweight(\Leaf(\idx))=\weight_{\idx},\qquad
\witnessweight(\Cat(\witnessU,\witnessV))=\witnessweight(\witnessU)+\witnessweight(\witnessV). 
\]
Thus concatenation takes $O(1)$ time and space, without copying or balancing the represented sequences. 
Binary reconstruction witnesses of this form also appear in~\cite[Appendix~B, proof of Theorem~1.3]{GGK25}. 
These operations apply to arbitrary index sequences. 
The feasibility of the constructed witnesses is proved in \cref{sec:construction}. 

All function trees use the same balanced binary subdivision of $[0,\inputsize)$. 
An interval $[\blockleft,\blockright)$ of length at least two splits at $\midpoint=\lfloor(\blockleft+\blockright)/2\rfloor$. 
There is one leaf per coordinate, and the height is at most $\lceil\log\inputsize\rceil$. 
Intervals are carried by recursive calls rather than stored in nodes. 
We identify a root with its represented function when no confusion arises. 

A function node stores two child pointers, a suffix handle $\suffix$, and the score $\last$ at its rightmost coordinate.
At a leaf, it represents the pair $(\witnessweight(\suffix),\, \suffix)$.
At an internal node, a selected child pair $(\scalar,\witnessU)$ gives the pair
\begin{equation}
\label{eq:node-semantics}
\bigl(\scalar+\witnessweight(\suffix),\ \Cat(\witnessU,\suffix)\bigr).
\end{equation}
The concatenation is conceptual. 
Reading a node does not allocate a witness. 
The cache satisfies
\[
\last = 
\begin{cases}
    \witnessweight(\suffix), & \text{at a leaf}, \\
    \last(\text{right child}) + \witnessweight(\suffix), & \text{at an internal node}. 
\end{cases}
\]
Thus the weight of the suffix is also the additive lazy tag. 
Initially, we build one full tree for $\envelope_{\emptyset}$ with empty suffixes and zero caches, in $O(\inputsize)$ time and space. 
Every version retains this subdivision, including on constant intervals. 

\paragraph{Basic operations.}
The operation $\Append(\noderoot,\witnessV)$ appends $\witnessV$ at every coordinate of $\noderoot$.
It copies the root, replaces its suffix by $\Cat(\suffix,\witnessV)$, and increases its cache by $\witnessweight(\witnessV)$.
It takes $O(1)$ time and creates at most one function node and one witness node.

For a nonleaf root, $\Children(\noderoot)$ returns its two restrictions by applying $\Append$ to each child with the parent's suffix.
The order is child suffix followed by parent suffix, as in \eqref{eq:node-semantics}.
Conversely, $\JoinTrees(\funleft,\funright)$ creates a parent with children $\funleft,\funright$, empty suffix, and cache $\last(\funright)$.
Both operations take $O(1)$ time and create $O(1)$ nodes.
They leave all existing nodes unchanged.

\subsection{Prefix updates}

For $\cut\in[0,\inputsize]$, let $\PrefixAppend(\noderoot,\cut,\witnessV)$ append $\witnessV$ at coordinates below $\cut$ and leave all other coordinates unchanged.
This is an ordinary persistent prefix update, given in \cref{alg:prefix-append}.
On a fully covered interval it calls $\Append$; on an uncovered interval it returns the original root.
Only the interval containing the boundary is partially covered at any level.
Each level performs $O(1)$ operations, so the update takes $O(\log\inputsize)$ time and creates $O(\log\inputsize)$ function and witness nodes.

\begin{algorithm}[t]
\caption{$\PrefixAppend(\noderoot,\cut,\witnessV)$ on $[\blockleft,\blockright)$.}
\label{alg:prefix-append}
\begin{algorithmic}[1]
\setlength{\itemsep}{2pt}
\vspace{1mm}
\State \textbf{if} $\cut\leq\blockleft$ \textbf{then} \Return $\noderoot$
\State \textbf{if} $\blockright\leq\cut$ \textbf{then} \Return $\Append(\noderoot,\witnessV)$
\State $(\funleft,\funright)\gets\Children(\noderoot)$
\If{$\cut\leq\left\lfloor\dfrac{\blockleft + \blockright}{2}\right\rfloor$}
  \State \Return $\JoinTrees(\PrefixAppend(\funleft,\cut,\witnessV),\funright)$
\Else
  \State \Return $\JoinTrees(\Append(\funleft,\witnessV),\PrefixAppend(\funright,\cut,\witnessV))$
\EndIf
\end{algorithmic}
\end{algorithm}

Since $\predenv_{\idx}(\leftbound) = 0$ with an empty witness for $\leftbound > \idx$, the recurrence in \eqref{eq:dp} is realized by a single prefix update: 
\begin{equation}
\label{eq:prefix-dp}
\dpopt_{\idx} = \PrefixAppend(\predenv_{\idx},\  \idx + 1,\  \Leaf(\idx)). 
\end{equation}

\subsection{Merging envelopes}

\begin{lemma}
\label{lem:merge}
Let $\funleft$ and $\funright$ be represented on the same interval $[\blockleft, \blockright)$. 
Suppose $\funright(\leftbound) > \funleft(\leftbound)$ holds exactly on a prefix of this interval. 
Their pointwise maximum can be represented in $O(\log \inputsize)$ time using $O(\log \inputsize)$ new function and witness nodes. 
At equality, the witness of $\funleft$ is chosen. 
Both input trees remain unchanged. 
\end{lemma}
\begin{proof}
Consider \cref{alg:merge}. 
At a leaf, the strict comparison selects the larger score, preferring $\funleft$ at equality. 
Otherwise, $\Children$ incorporates the parent suffixes into the child roots. 
If $\last(\funright_{\mathrm L})>\last(\funleft_{\mathrm L})$, the prefix property implies that $\funright$ exceeds $\funleft$ throughout the left half. 
Otherwise, $\funright$ does not exceed $\funleft$ anywhere in the right half. 
Thus one half is selected directly, and only the other requires recursion. 
Restriction preserves the prefix property, proving correctness by induction. 

Each call makes at most one recursive call and otherwise takes $O(1)$ time and creates $O(1)$ new function and witness nodes. 
The tree height gives the claimed bounds. 
\end{proof}

\begin{algorithm}[t]
\caption{
$\Merge(\funleft, \funright)$ for functions satisfying the prefix property. 
}
\label{alg:merge}
\begin{algorithmic}[1]
\setlength{\itemsep}{2pt}
\vspace{1mm}
\If{the interval is a singleton}
  \If{$\last(\funright)>\last(\funleft)$}
    \State \Return $\funright$
  \Else
    \State \Return $\funleft$
  \EndIf
\EndIf
\State $(\funleft_{\mathrm L},\funleft_{\mathrm R})\gets\Children(\funleft)$
\State $(\funright_{\mathrm L},\funright_{\mathrm R})\gets\Children(\funright)$
\If{$\last(\funright_{\mathrm L})>\last(\funleft_{\mathrm L})$}
  \State \Return $\JoinTrees(\funright_{\mathrm L},\Merge(\funleft_{\mathrm R},\funright_{\mathrm R}))$
\Else
  \State \Return $\JoinTrees(\Merge(\funleft_{\mathrm L},\funright_{\mathrm L}),\funleft_{\mathrm R})$
\EndIf
\end{algorithmic}
\end{algorithm}

Cached scores allow $\Merge$ to construct the pointwise maximum in a single descent. 
By \cref{lem:envelopes}, $\Merge(\envelope_{\lowset},\envelope_{\highset})$ is applicable whenever the endpoint sets are separated in this value order.

\subsection{Evaluation and reporting}

\begin{lemma}
\label{lem:witnesses}
A represented score at a given coordinate can be returned in $O(\log\inputsize)$ time, and its witness of length $\outputsize$ can be reported in $O(\outputsize+\log\inputsize)$ time.
Neither query changes the stored structure.
\end{lemma}
\begin{proof}
Descend to the coordinate's leaf, collecting the suffixes $\suffix_0,\dots,\suffix_{\pathheight}$ in root-to-leaf order.
By \eqref{eq:node-semantics}, the score is $\sum_{\jdx=0}^{\pathheight}\witnessweight(\suffix_{\jdx})$, which is computed in $O(\log\inputsize)$ time.
The witness is the concatenation $\suffix_{\pathheight}\suffix_{\pathheight-1}\cdots\suffix_0$. 
Discard empty handles and expand the remaining forest in this order, visiting the left child of each concatenation before the right child.
Every concatenation has two nonempty children.
The unfolded forest with $\outputsize$ leaf occurrences therefore has at most $2\outputsize$ node occurrences, regardless of sharing.
Expansion takes $O(\outputsize)$ time.
The traversal uses $O(\outputsize+\log\inputsize)$ workspace including the output and allocates no persistent nodes.
\end{proof}

\section{Preprocessing and queries}
\label{sec:construction}

We now prove \cref{thm:main}.
Maintain an ordinary mutable segment tree, denoted $\outertree$, on the value domain $[0,\inputsize)$.
Its entries are roots of the persistent trees from \cref{sec:persistence}, which are indexed by the left boundary. 
For an outer node $\treenode$ with value interval $\nodeinterval_{\treenode}$, maintain, immediately before processing position $\idx$,
\begin{equation}
\label{eq:invariant}
\noderoot_{\treenode}
=\envelope_{\{\jdx<\idx\, \mid\, \seqval_{\jdx}\in\nodeinterval_{\treenode}\}}.
\end{equation}
For positive scores, the stored witness is an optimum subsequence ending in this set. 
For zero scores, it is empty. 
Initially every entry points to the all-zero tree for $\envelope_{\emptyset}$.

The outer tree uses $\Merge$ as its combine operation.
For scores, this operation is pointwise maximum, with identity $\envelope_{\emptyset}$.
Its $O(\log\inputsize)$ implementation requires the ordered endpoint sets of \cref{lem:envelopes}.
Accordingly, $\outertree.\mathsf{prod}(0,\seqval_{\idx})$ folds the canonical intervals in increasing value order, placing the accumulated lower-value envelope first in each merge.
A point assignment $\outertree.\mathsf{set}(\seqval_{\idx},\dpopt_{\idx})$ recomputes ancestors by merging the left child's root with the right child's root.
Both operations use $O(\log\inputsize)$ merges, all on value-separated endpoint sets.
The operation $\outertree.\mathsf{prod}()$ returns the outer root's stored envelope.
The outer subdivision is the usual balanced recursive subdivision, so no padding is needed.

\begin{algorithm}[t]
\caption{Preprocessing the relabeled sequence.}
\label{alg:build}
\begin{algorithmic}[1]
\setlength{\itemsep}{2pt}
\vspace{1mm}
\State Build the full persistent tree for $\envelope_{\emptyset}$
\State Initialize all entries of $\outertree$ to $\envelope_{\emptyset}$
\State $\rootptr[0]\gets\envelope_{\emptyset}$
\For{$\idx=0,\dots,\inputsize-1$}
  \State $\predenv_{\idx}\gets\outertree.\mathsf{prod}(0,\seqval_{\idx})$
  \State $\dpopt_{\idx}\gets\PrefixAppend(\predenv_{\idx},\idx+1,\Leaf(\idx))$
  \State $\outertree.\mathsf{set}(\seqval_{\idx},\dpopt_{\idx})$
  \State $\rootptr[\idx+1]\gets\outertree.\mathsf{prod}()$
\EndFor
\end{algorithmic}
\end{algorithm}

\begin{lemma}
\label{lem:construction}
\Cref{alg:build} constructs $\rootptr[\rightbound]=\envelope_{[0,\rightbound)}$ for every $\rightbound\in[0,\inputsize]$, with optimum witnesses for positive values and empty witnesses for zero values.
Its time and total allocated space are $O(\inputsize\log^2\inputsize)$.
\end{lemma}
\begin{proof}
We prove \eqref{eq:invariant} and its witness assertion by induction on $\idx$.
They hold initially.
The canonical intervals partition $[0, \seqval_{\idx})$. 
Their ordered merges yield $\predenv_{\idx}$ by \cref{lem:envelopes,lem:merge}. 
For $\leftbound\leq\idx$, a predecessor witness of positive weight ends at an earlier position of smaller value, and appending $\idx$ gives a feasible subsequence of weight $\predenv_{\idx}(\leftbound)+\weight_{\idx}$.
A predecessor witness of weight zero is empty and gives the singleton $\idx$.
For $\leftbound>\idx$, the update leaves the zero value and empty witness unchanged.
Thus \eqref{eq:dp} proves that the new root represents $\dpopt_{\idx}$ with the required witnesses.

Since the ranks are distinct, the leaf for $\seqval_{\idx}$ was previously empty. 
At each updated ancestor, the child endpoint sets are separated by value. 
They partition the set in \eqref{eq:invariant} with $\idx$ replaced by $\idx + 1$. 
Their merge preserves the invariant and selects an optimum witness.
All other outer nodes are unchanged.
The outer root consequently represents $\envelope_{[0,\idx+1)}$.
Persistence preserves each saved root and its witnesses.

Each iteration uses $O(\log\inputsize)$ merges and one prefix append.
By \cref{lem:merge} and the prefix-update bound, it takes $O(\log^2\inputsize)$ time and allocates $O(\log^2\inputsize)$ function and witness nodes.
Relabeling takes $O(\inputsize\log\inputsize)$ time, and building the initial trees takes $O(\inputsize)$ time and space.
The input, mutable outer tree, and saved roots occupy $O(\inputsize)$ space.
Every allocated node has constant size, so even retaining all temporary nodes uses $O(\inputsize\log^2\inputsize)$ space.
The recursion stack uses $O(\log\inputsize)$ additional space. 
\end{proof}

\paragraph{Queries.}
For a nonempty interval $[\leftbound,\rightbound)$, evaluate $\rootptr[\rightbound]$ at $\leftbound$.
By \eqref{eq:query} and \cref{lem:construction}, its score is $\rangeopt(\leftbound,\rightbound)$, and its witness is an optimum increasing subsequence in the query interval.
\Cref{lem:witnesses} gives $O(\log \inputsize)$ value-query time and $O(\outputsize + \log\inputsize)$ reporting time. 
An empty interval returns zero and an empty list directly.
This proves \cref{thm:main}.

\section*{AI Assistance Disclosure}
ChatGPT (OpenAI) assisted with literature searches, algorithm design, proof development and checking, counterexample searches, implementation, and manuscript preparation. 
The author is responsible for the statements and proofs in this paper. 

\begingroup
\setlength{\emergencystretch}{3em}
\bibliographystyle{alphaurl}
\bibliography{references}
\endgroup
\end{document}